\documentclass[sigconf]{acmart} 

\AtBeginDocument{%
  \providecommand\BibTeX{{%
    \normalfont B\kern-0.5em{\scshape i\kern-0.25em b}\kern-0.8em\TeX}}}

\usepackage{amsmath,amssymb,amsfonts}
\usepackage{algorithmic}
\usepackage{multirow}
\usepackage{siunitx}

\newtheorem{theorem}{Theorem}
\newtheorem{lemma}[theorem]{Lemma}

\newcommand{\Mat}[1]{\mathbf{#1}}

\newcommand{\fullname}{\emph{Fair} \emph{G}raph based Rec\emph{O}mmendation}
\newcommand{\shortname}{\emph{FairGo} }

\usepackage{booktabs}
\usepackage{subfig}
\usepackage{booktabs}
\usepackage[export]{adjustbox}

\copyrightyear{2021}
\acmYear{2021}
\setcopyright{iw3c2w3}
\acmConference[WWW '21]{Proceedings of the Web Conference 2021}{April 19--23, 2021}{Ljubljana, Slovenia}
\acmBooktitle{Proceedings of the Web Conference 2021 (WWW '21), April 19--23, 2021, Ljubljana, Slovenia}
\acmPrice{}
\acmDOI{10.1145/3442381.3450015}
\acmISBN{978-1-4503-8312-7/21/04}

\begin{document}

\vspace{-0.2cm}
\title{Learning Fair Representations for Recommendation: \\A Graph-based Perspective }

\author{Le Wu}
\affiliation{\institution{Hefei University of Technology}}
\affiliation{\institution{Institute of Artificial Intelligence, Hefei Comprehensive National Science Center}}
\affiliation{\institution{Intelligent Interconnected Systems Laboratory of Anhui Province
}}
\email{lewu.ustc@gmail.com} 

\author{Lei Chen}
\affiliation{\institution{Hefei University of Technology}}
\affiliation{\institution{Intelligent Interconnected Systems Laboratory of Anhui Province
}}
\authornotemark[1] 
\email{chenlei.hfut@gmail.com}

\author{Pengyang Shao}
\affiliation{\institution{Hefei University of Technology}}
\affiliation{\institution{Intelligent Interconnected Systems Laboratory of Anhui Province
}}
\email{shaopymark@gmail.com}

\author{Richang Hong}
\affiliation{\institution{Hefei University of Technology}} 
\affiliation{\institution{Intelligent Interconnected Systems Laboratory of Anhui Province}} 
\email{hongrc.hfut@gmail.com}

\author{Xiting Wang}
\affiliation{\institution{Microsoft Research}}
\email{xitwan@microsoft.com}

\author{Meng Wang}
\authornote{Lei Chen and Meng Wang are corresponding authors.}
\affiliation{\institution{Hefei University of Technology}}
\affiliation{\institution{Institute of Artificial Intelligence, Hefei Comprehensive National Science Center}}
\affiliation{\institution{Intelligent Interconnected Systems Laboratory of Anhui Province}}  
\email{eric.mengwang@gmail.com}

\vspace{-0.3cm}
\begin{abstract}
As a key application of artificial intelligence, recommender systems are among the most pervasive computer aided systems to help users find potential items of interests. Recently, researchers paid considerable attention to fairness issues for artificial intelligence applications. Most of these approaches assumed independence of instances, and designed sophisticated models to eliminate the sensitive information to facilitate fairness. However, recommender systems differ greatly from these approaches as users and items naturally form a user-item bipartite graph, and are collaboratively correlated in the graph structure. In this paper, we propose a novel graph based technique for ensuring fairness of any recommendation models. Here, the fairness requirements refer to not exposing sensitive feature set in the user modeling process. Specifically, given the original embeddings from any recommendation models, we learn a composition of filters that transform each user's and each item's original embeddings into a filtered embedding space based on the sensitive feature set. For each user, this transformation is achieved under the adversarial learning of a user-centric graph, in order to obfuscate each sensitive feature between both the filtered user embedding and the sub graph structures of this user. Finally, extensive experimental results clearly show the effectiveness of our proposed model for fair recommendation.
We publish the source code at https://github.com/newlei/FairGo.

\end{abstract}



\maketitle

\section{Introduction}
\begin{small}
\begin{figure*} [htb]
    \centering 
    \subfloat[Data processing step]{           
        \includegraphics[width=0.4\linewidth,valign=t]{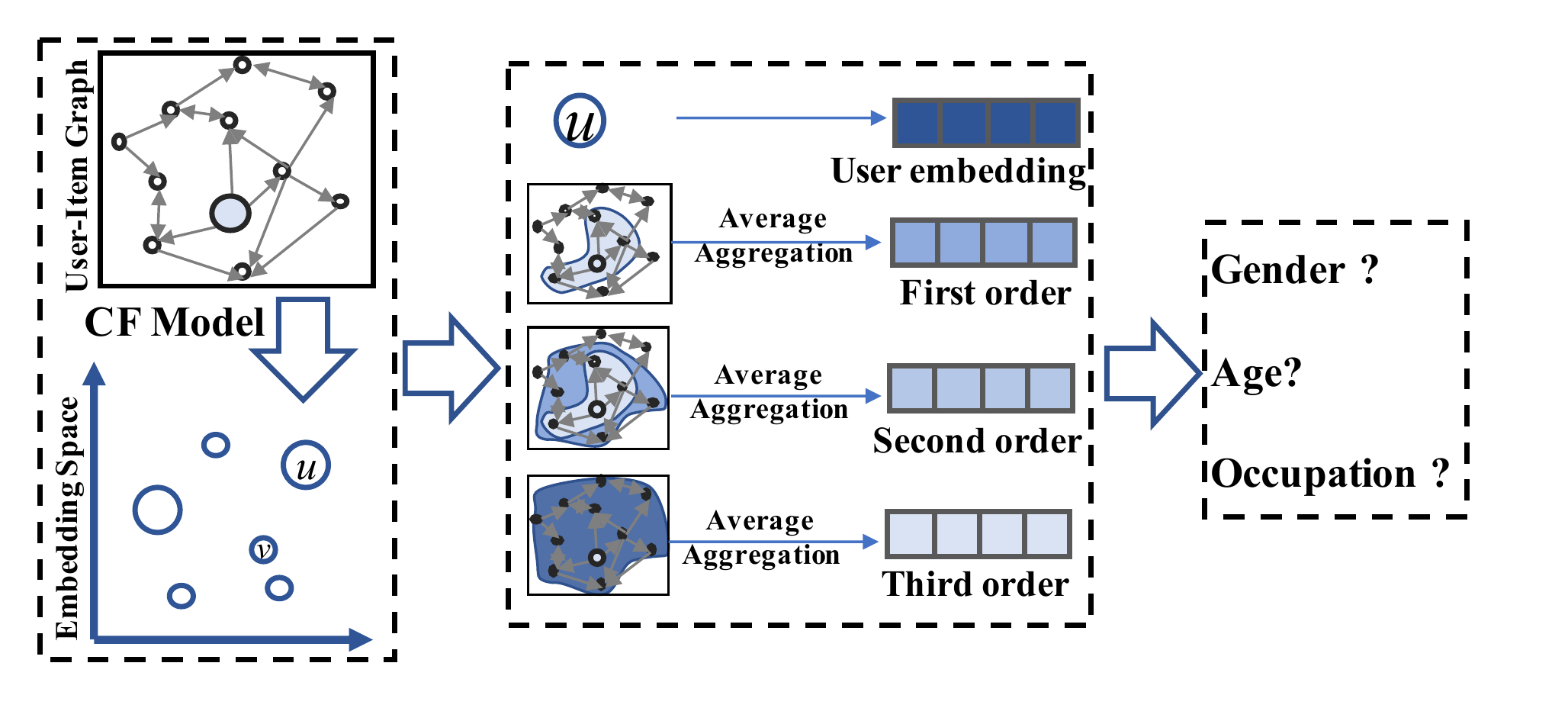}  
        \label{fig:intro_example}          

        }\ 
    \subfloat[Attribute prediction performance]{
        \adjustbox{width=0.5\linewidth,valign=t}{\begin{tabular}{|l|l|l|l|l|}
\hline
\multirow{2}{*}{Model} & \multirow{2}{*}{Input}& \multicolumn{3}{l|}{Sensitive attribute prediction performance} \\ \cline{3-5} 
& & Gender(AUC) & Age(F1) & Occupation(F1) \\ \hline
\multirow{4}{*}{PMF} & User embedding & 0.6615 & 0.3821 & 0.1332 \\ \cline{2-5} 
 & First order & 0.6181 & 0.3569 & 0.1407 \\ \cline{2-5} 
 & Second order & 0.5102 & 0.3431 & 0.1405 \\ \cline{2-5} 
 & Third order & 0.5004 & 0.3234 & 0.1289 \\ \hline
\multirow{4}{*}{GCN} & User embedding & 0.7041 & 0.4215 & 0.1485 \\ \cline{2-5} 
 & First order & 0.6804 & 0.3782 & 0.1474 \\ \cline{2-5} 
 & Second order & 0.5811 & 0.3509 & 0.1418 \\ \cline{2-5} 
 & Third order & 0.5129 & 0.3449 & 0.1296 \\ \hline
\end{tabular}
        \label{tab:primilary}}}
    \vspace{-0.2cm}
    \caption{Performance of two recommendation models~(PMF~\cite{mnih2008PMF} and GCN~\cite{AAAI2020revi}) for sensitive attribute prediction on MovieLens dataset. After learning user and item embeddings, we then extract the $l$-th user-centric subgraph embedding of each user. The learned $l$-th order embedding vector is treated as feature input for sensitive attribute prediction. We observe that each $l$-th order user-centric graph representation is helpful for attribute prediction. Details  can be found in the experiments.}
    \vspace{-0.3cm}
\end{figure*}
\end{small}

With the information explosion, recommender systems have been widely deployed in most platforms and have penetrated our daily life~\cite{recsys2016deep,KDD2018PinSage,koren2009MF,wu2020joint,AAAI2020revi}. These systems shape the news we consume, the movie we watch, the restaurant we choose, the job we seek and so on. While recommendation systems could better help users to find potentially interesting items, the recommendation results are also vulnerable to biases and unfairness. E.g.,  current recommendation results are empirically shown to favor a particular demographic group over others~\cite{ekstrand2018all,ekstrand2018exploring}. Career recommendation shows apparent gender-based discrimination even for equally qualified men and women~\cite{lambrecht2018algorithmic}. Ad recommendation results display racial biases between users with similar preferences~\cite{sweeney2013discrimination}. 

As biases in algorithms have been ubiquitous in these human centric  artificial intelligence applications, how to evaluate and improve algorithmic fairness to  benefit all users has become a hot research topic~\cite{pedreshi2008discrimination,NIPS2016equality}. Given a specific sensitive attribute, researchers have designed metrics for measuring fairness in supervised settings~\cite{dwork2012fairness,NIPS2016equality}. These metrics encourage the proportion  of sensitive attribute values in a protected group classified as positive is identical to that of the unprotected group~\cite{NIPS2016equality,ICLR2016censoring}.
Among all debiasing models, fair representation learning has become very popular
and widely studied due to the simplicity, generality and the advances of representation learning techniques~\cite{zemel2013learning,ICLR2016censoring,beutel2017data,nips2017controllable,bose2019compositional}.  These fair representation learning approaches learn data representations to maintain the main task while filtering any sensitive information hidden in the data representations. The fairness requirements are achieved by specific fairness regularization terms~\cite{zemel2013learning,yao2017beyond,cikm2018fairness}, or relied on adversarial learning techniques~\cite{goodfellow2014generative} that try to match the conditional distribution of representations given each sensitive attribute value to be identical~\cite{ICLR2016censoring,beutel2017data,nips2017controllable,bose2019compositional}.

In this paper, we focus on fair representation learning for fair recommendation, which tries to  eliminate sensitive information in the representation learning~\cite{zemel2013learning,beutel2017data,bose2019compositional}. Here, the fairness requirements refer to the fact that recommender systems do not expose any sensitive user attribute, such as gender, occupation.  In fact, state-of-the-art recommender systems rely on learning user and item embeddings for recommendation. E.g.,  the popular latent factor models learn free user and item embeddings for recommendation~\cite{mnih2008PMF,UAI2009BPR}. Recently, researchers argued that users and items naturally form a user-item bipartite graph structure~\cite{wu2020joint,wu2019neural}, and neural graph based models learn the user and item embeddings by injecting the graph structure in user and item embedding process, then receive state-of-the-art recommendation performance~\cite{KDD2018PinSage,AAAI2020revi}.
As learning user and item representations have become the key building block for modern recommender systems, we also  focus on learning fair user and item embeddings, such that the fair representation learning could be integrated into modern recommendation architecture.
In other words, the fair recommendation problem turns to learning fair user and item representations, such that any sensitive information could not be exposed from the learned embeddings.

In fact, even the user-item interaction behavior do not explicitly contain any user sensitive information, directly applying state-of-the-art user and item representation learning would lead to user sensitive information leakage, due to the widely supported correlation between user behavior and her attributes in social theories~\cite{kosinski2013private,wu2019neural,wu2017modeling}. E.g., a large scale study shows that users' private traits~(e.g., gender, political views) are predictable from their like behaviors from \emph{Facebook}. Therefore, a naive idea is  to borrow the current fairness-aware supervised machine learning techniques to ensure fairness on the user embeddings. This solution  alleviates unfairness of user representation learning to some extend. However, we argue that it is still far from satisfactory due to the uniqueness of the recommendation problem. Most fairness based machine learning tasks assume independence of entities, and eliminate unfairness of each entity independently without modeling the correlations with other entities. In recommender systems, users and items naturally form a user-item bipartite graph, and are collaboratively correlated in the systems. 
In these systems, each user's embedding is not only related to her own behavior, but also implicitly correlated with similar users' behaviors, or the user's behavior on similar items.  The collaborative correlation between users break the independence assumption in previous fairness based models, and is the foundation of collaborative filtering based recommendation. As such, even though a user's sensitive attributes are eliminated from her embedding, the user-centric structure may expose her sensitive attribute and lead to unfairness. To validate this assumption, we show an example of how a user's attribute can be inferred from the local graph structure of this user with state-of-the-art embedding models.  It can be observed from Figure~\ref{tab:primilary} that the attributes of users are not only exposed through her embedding, but also through surrounding neighbors' embeddings. This preliminary study empirically shows that each user's sensitive attributes are also related to the user-centric graph. As users and items form a graph structure, it is important to learn fair representations for recommendation from a graph based perspective.

To this end, in this paper, we propose a graph based perspective for fairness aware representation learning of any recommendation models. We argue that as the recommendation models are diversified and complicated in the real production environment, the proposed model should better be model-agnostic. By defining a sensitive feature set, our proposed model takes the user and item embeddings from any recommendation models as input, and learns a filter space such to obfuscate any sensitive information in the sensitive attribute set, while simultaneously maintains recommendation accuracy. Specifically, we learn a composition of each sensitive attribute filter that transforms each user's and item's original embeddings into a filtered embedding space. As each user can be represented as an ego-centric graph structure, the filters are learned under a graph based adversarial training process. Each discriminator tries to predict the corresponding attribute, and the filters are trained to eliminate any sensitive information that may be exposed by the user-centric graph structure. Finally, we perform extensive experimental results on two real-world datasets with varying sensitive information. The results clearly show the effectiveness of our proposed model for fair recommendation.

\section{Related Work}

\textbf {Recommendation Algorithms.} 
In a recommender system, there are two sets of entities:  a user set {\small$U$~($|U|\!=\!M$)}, and an item set {\small $V$~($|V|\!=\!N$)}.
Users interact with items to form a user-item interaction matrix $\Mat{R}\in \mathbb{R}^{M\times N}$. If user~$u$ has rated item~$v$, then ~$r_{uv}$ is the detailed rating value, otherwise
~$r_{uv}=0$. Naturally, we could formulate a user-item bipartite graph as
$\mathcal{G}=<{U}\cup {V}, \mathbf{A}>$, with $\mathbf{A}$ is formulated based on the rating matrix $\Mat{R}$ as:

\vspace{-0.2cm}
\begin{equation} \label{eq:cfmat2graph}
\mathbf{A}=\left [
    \begin{array}{c c}\Mat{R}\quad& \Mat{0}^{N\times M} \\
    \Mat{0}^{M\times N}\quad&\Mat{R^T} 	
    \end{array}
    \right].
\end{equation}

Learning high quality user and item embeddings has become the building block for successful recommender systems~\cite{mnih2008PMF,SIGIR2019NGCF,AAAI2020revi,wu2020joint}. Let {~\small$\Mat{E}\in \mathbb{R}^{D\times(M+N)}$} denote the embeddings of users and items learned by a recommendation $Enc$: $\Mat{E}=Enc(\mathcal{G})=[\mathbf{e}_1, ...,\mathbf{e}_u,..., \mathbf{e}_v,...\mathbf{e}_{M+N}]$.  After that, the predicted preference $\hat{r}_{uv}$ of user $u$ to item $v$ is calculated as the inner product between the corresponding user and item embeddings as: $\hat{r}_{uv}=\mathbf{e}^T_u\times \mathbf{e}_v$.

Currently, there are two classes of embedding approaches: the classical latent factor based models~\cite{mnih2008PMF,UAI2009BPR} and neural graph based models~\cite{SIGIR2019NGCF,AAAI2020revi}. Latent factor models adopt matrix factorization approaches to learn the free user and item ID embeddings. In contrast, the neural graph based models iteratively stack multiple graph convolution layers for node embedding in this user-item graph.  At each iteration $l+1$, each node's embedding at this layer is a convolution neighborhood aggregation by neighborhood's embeddings at layer $l$. Empirically, these neural graph based models show better performance by injecting the collaborative signal hidden in the graph for user and item enbeddubg learning~\cite{SIGIR2019NGCF,AAAI2020revi}.

\textbf{Algorithmic Fairness and Applications.}
As machine learning and data mining are widely applied for knowledge discovery to guide automated decision making, there is much interest in discovering, measuring and ensuring fairness~\cite{pedreshi2008discrimination,icml2018learning,kdd2019pairwisefairness}. Among all fairness metrics, group fairness is widely used to measure how the underrepresented group is treated in this process~\cite{NIPS2016equality}. 
Current solutions for fairness requirements can be classified into causal based approaches~\cite{nips2017counterfactual,khademi2019fairness}, ranking based models~\cite{kdd2019pairwisefairness}, and fair representation learning models~\cite{zemel2013learning,ICLR2016censoring,icml2018learning,bose2019compositional}. In this paper, we focus on fair representation learning due to its generality  and the recent rapid progress of representation learning techniques. Fair representation learning models either added fairness-based regularization terms~\cite{zemel2013learning,cikm2018fairness,yao2017beyond} in the objective function or relied on the adversarial learning models to ensure group fairness~\cite{icml2018learning,bose2019compositional}. 
Borrowing the success of GANs~\cite{goodfellow2014generative}, adversarial fair representation models have a feature learning module and an additional discriminator to guess the sensitive information. These two parts play a minimax game, and adversarial upper bounds on group fairness metrics can be achieved~\cite{icml2018learning}. Compared to the manually defined fairness regularization terms, adversarial training for fairness shows the theoretical elegance and the learned representations can be transferred for many downstream tasks. Most of the current fairness representation learning focused on binary supervised tasks. A recent work tackled the problem of learning fair representation learning from graph~\cite{bose2019compositional}. This approach advanced previous works with state-of-the-art graph embedding based representation learning models, and a composition of discriminators for modeling the correlation of sensitive features~\cite{bose2019compositional}. However, the graph structure is only utilized for accurate node embedding learning, and the fairness is still achieved by independently filtering out each node's sensitive information. 
E.g., in recommender systems with user-item bipartitie graph, this model may lead to unfairness as users' sensitive information 
is exposed by the items they like.

\textbf{Recommendation Fairness.}
In recommender systems, researchers observed popularity and demographic disparity of the current user-centric applications and recommender systems, with different demographic groups obtain different utility from the recommender systems~\cite{ekstrand2018all,ekstrand2018exploring,lambrecht2018algorithmic}. Researchers empirically showed that, the post-processing technique that improves recommendation diversity would amplify user unfairness~\cite{WWW2018user}.
Researchers proposed four new metrics for collaborative filtering based recommendation with a binary sensitive attribute, in order to measure the discrepancy between the prediction behavior for disadvantages users and advantaged users. Theses fairness metrics are treated as fairness regularization terms for group fairness in recommendation~\cite{yao2017beyond}. A fairness aware tensor based recommendation is proposed by isolating sensitive attributes in the latent factor matrix, and the remaining features are regularized to keep away from sensitive attributes~\cite{cikm2018fairness}. Instead of directly debiasing results in the model learning process, re-ranking models are also applied in search and recommendation systems, with the well designed fairness metrics to guide the learning to process to mitigate the disparity~\cite{kdd2019pairwisefairness,kdd2019fairnessreranking}.  Some studies tried to find the casual effect or design explainable models for users' behaviors, in order to  ensure fairness grounded on causal effect or explainable components~\cite{khademi2019fairness}.  We  differ from these recommendation fairness models as we argue that users naturally form a user-item bipartite graph structure, and users' sensitive information can be exposed from her local graph structure. Therefore,
we consider the fairness issue from a graph based perspective.

\section{The Proposed \shortname Model}

Most recommender systems are based on embedding based models, and can be very complex and time-consuming due to the large volume of users and heterogeneous data~\cite{KDD2018PinSage,recsys2016deep}. Therefore, the user and item embedding learning process are performed offline, and it is nearly impossible to retrain the embedding models from time to time. We attempt to design a model that takes user and item embeddings from any recommendation model as input, i.e., {\small$\mathbf{E}$}, and our goal is to achieve a model-agnostic based fair representation learning in the filtered space. Here, the fairness requirements refer to a protected or sensitive user attribute set $\mathbf{X}\in\mathbb{R}^{K*M}$ with $K$ sensitive attributes~(e.g., gender and age), which are not encouraged to help recommendation. In the following, we introduce our proposed \fullname~(FairGo) model for fairness requirements in recommender systems, followed by the theoretical analysis.

\begin{small}
\begin{figure*} [htb]
  \begin{center}
    \includegraphics[width=140mm]{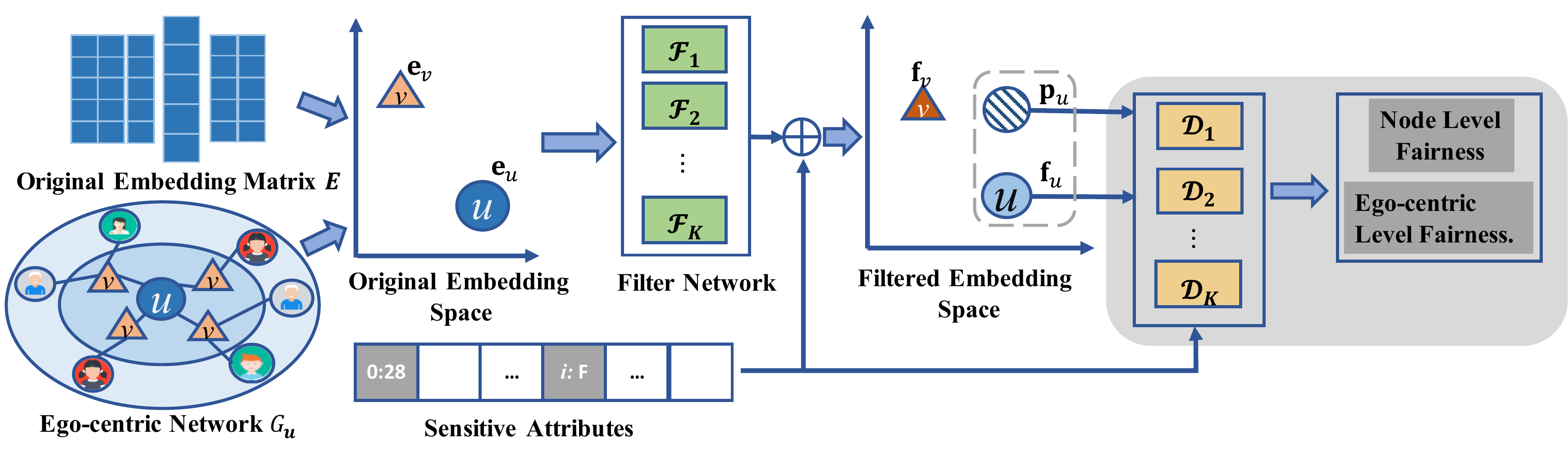}
  \end{center} 
  \caption{The overall structure of our proposed \shortname model.}\label{fig:overall structure} 
\end{figure*}
\end{small}

\subsection{Overall Architecture}

Given the original embedding matrix {\small$\mathbf{E}$} and the sensitive attributes {\small$\mathbf{X}$}, \shortname designs a combination of $K$ sub filters as the filter structure $\mathcal{F}$ to remove information about the user protected attributes $\mathbf{X}$, such that each node~(user node or item node) is filtered from the original embedding space {\small$\mathbf{E}$} to a filtered embedding space  as: {\small$\mathbf{F}=\mathcal{F}(G,\mathbf{E},\mathbf{X})$}, with $\mathbf{F}=[\mathbf{F}_U, \mathbf{F}_V]\in \mathbb{R}^{D\times (M+N)}$ .  As there are $K$ sensitive attributes,  the filter network $\mathcal{F}$ is composed of $K$ sub filters as: {\small$\mathcal{F}=[\mathcal{F}^k]_{k=1}^{K}$}, with each sensitive attribute $k$ is associated with a sub filter $\mathcal{F}^k$.
Then,  each entity~(user or item) is filtered and represented in the filtered embedding space as: 

\vspace{-0.2cm}
\begin{small}
\begin{equation} \label{eq:filter_pred}
\mathbf{f}_{i}=\frac{\sum_{k=1}^{K} \mathcal{F}^k(\mathbf{e}_i)}{K}.
\end{equation}
\end{small}
\vspace{-0.2cm}

Given the filtered embedding space, the predicted preference $\hat{r}_{uv}$ of user $u$ to item $v$ is calculated as:
\begin{small}
\begin{equation}\label{eq:filter_pred_r}
\hat{r}_{uv}=\mathbf{f}^T_u\times\mathbf{f}_v.
\end{equation}
\end{small}
\vspace{-0.3cm}

With the overall filter network structure to filter original embeddings in a filter space, we argue that the fairness-aware recommender systems need to satisfy two goals: \emph{representative} for users' personalized preferences while \emph{fair} for the sensitive attributes. On one hand, the filtered embeddings should be representative of 
users' preferences to facilitate recommendation accuracy. On the other hand, these filtered embeddings should be fair and do not leak any information that correlates to each user's personalized sensitive information.  

In this paper, we adopt adversary training techniques to achieve fairness. Specifically, given the filtered networks $[\mathcal{F}^k]_{k=1}^K$, there are $K$ \emph{d}iscriminator sub networks. By taking  the filtered embedding~$\mathbf{f}_u$ as input, the $k$-th sub discriminator attempts to predict the value of the $k$-th sensitive attribute.
In other words, each sub discriminator $\mathcal{D}^k$ works as a classifier to guess the $k$-th attribute. 
The filter network and the discriminator network play the following two-player minimax game with the following value function $V(\mathcal{F},\mathcal{D})$:

\vspace{-0.3cm}
\begin{flalign} \label{eq:original_likelihood}
&\arg\underset{\mathcal{F}}{\max }~\arg\underset{\mathcal{D}}{\min}V (\mathcal{F},\mathcal{D})=V_{R}-\lambda V_{G} \\\nonumber
&=\mathbb{E}_{(u,v,r,x_u)\sim p(\mathbf{E},\mathbf{R},\mathbf{X})}
[\ln q_{\mathcal{R}}(r|(\mathbf{f}_u, \mathbf{f}_v)=\mathcal{F}(\mathbf{e}_u,\mathbf{e}_v))-\\
&\quad \lambda \ln q_{\mathcal{D}}(x|(\mathbf{f}_u, \mathbf{p}_u)=\mathcal{F}(\mathbf{e}_u,\mathbf{e}_v))],
\end{flalign}

\noindent where $V_{R}$ is the log likelihood of the rating distribution and $V_G$ is the log likelihood of the predicted attribute distribution. $\lambda$ is a balance parameter that balances these two value functions. When $\lambda$ equals zero, the fairness requirements disappear.

For the rating distribution, we assume it follows a Gaussian distribution, with the mean of the Gaussian distribution is the predicted rating as shown in Eq.\eqref{eq:filter_pred_r}. Therefore, the value function of 
$V_R$ is modeled as:  

\vspace{-0.2cm}
\begin{small}
\begin{equation} \label{eq:loss_r2}
V_{R}=-\sum_{u=1}^M\sum_{v=1}^V (r_{uv}-\hat{r}_{uv})^2,
\end{equation}
\end{small}
\vspace{-0.2cm}

\noindent  where the precision parameter in the Gaussian distribution is omitted as we can perform a reweight trick by tuning the balance parameter $\lambda$ of these two tasks.

\subsection{Graph based  Adversarial Learning for Fairness Modeling}
Given the sensitive attribute vector $\mathbf{x}_i$, a naive idea is to design the value function based on the current node's embedding as: 
\vspace{-0.2cm}
\begin{small}
\begin{equation} \label{eq:f1_loss}
V_{N}= \mathbb{E}_{(u,v,r,x_u)} \sum_{k=1}^{K } x_{uk}ln \mathcal{D}^k(\mathbf{f}_u).
\end{equation}
\end{small}
\vspace{-0.2cm}

In fact, the above value function only considers the fairness in the filtered embedding space with independence assumption of users. In recommender systems, users and items form a user-item bipartite graph. For each user $u$,  we use $G_u$ to denote the ego-centric network of user $u$ in the user-item graph $G$.  Specifically, the ego-centric network $G_u$ takes $u$ as the central node, and is a 
local neighborhood network that spans from $u$. With the ego-centric network $G_u$ of $u$,  the goal towards fairness requirements is that, $u$'s sensitive attribute is not exposed by her local network $G_u$.

The above Eq.\eqref{eq:f1_loss} simplifies the user-centric graph $G_u$ as a filtered node level representation, i.e., $\mathbf{f}_u$. Nevertheless, the independence assumption among users is not well supported in the user-item bipartite graph. In fact, the collaborative correlations between users are the foundation for building recommender systems. In the user-item  bipartite graph $G$, users are correlated through items in this graph structure. Trivial global representation $\mathbf{f}_u$ of each user $u$ may not well capture the local graph structure of this user. Therefore, given the filtered node embedding space, we also seek to obtain an ego-centric graph based structure representation of each user $u$ as:

\vspace{-0.3cm}
\begin{small}
 \begin{equation} \label{eq:graph_summary}
 \mathbf{p}_u=\mathcal{P}(G_u,\mathbf{F})=\mathcal{P}(G_u,\mathcal{F}(G,\mathbf{E},\mathbf{X}))),
 \end{equation}
\end{small}
\vspace{-0.2cm}

\noindent where  $\mathcal{P}$ is a structure representation function of the local graph summary of a user, and $\mathbf{p}_u$ is the output of the patch network that  summarizes user $u$ from her ego-centric graph structure $G_u$.E.g., it can be an aggregation of a user's up to $L$-th order neighborhood representation, or can be implemented with state-of-the-art sophisticated graph representation learning models~\cite{ying2018hierarchical}.

Similar as Eq.\eqref{eq:f1_loss}, given the local graph structure summary $\mathbf{p}_u$ of each user $u$, we also employ adversarial training to ensure each user's sensitive attributes are not exposed by local graph structure:

\vspace{-0.2cm}
\begin{small}
\begin{equation} \label{eq:f1_loss2}
V_{S}= \mathbb{E}_{(u,v,r,x_u)} \sum_{k=1}^{K } x_{uk} ln \mathcal{D}^k(\mathbf{p}_{u}).
\end{equation}
\end{small}
\vspace{-0.2cm}

As such, the fairness requirement is defined under the graph based adversarial learning process, with each user's ego-centric network structure is summarized as both the node-level based value function in as shown in  Eq.\eqref{eq:f1_loss} and the graph structure level function in Eq.\eqref{eq:f1_loss2}.  Then, the fairness based value function  $V_{G}$ is a combination of these two parts as: $V_{G}=V_{N}+V_{S}$, where the first part captures the node-level fairness, and the second part models the ego-centric fairness.

\subsubsection{Summary Network for Ego-centric Graph}
Now we focus on how to model ego-centric fairness based value function $V_{S}$. In other words, we need to build a summary network $\mathbf{p}_u$ for better ego-centric representation.

\textbf{Weighted Average Pooling.}
A simple yet effective implementation of the ego-centric graph summary structure as:

\vspace{-0.2cm}
\begin{small}
 \begin{equation}\label{eq:s_first_nei_mean}
 \mathbf{p}_u=\mathcal{P}(G_u,\mathbf{F})=\frac{\sum_{v\in A_u} r_{uv} \mathbf{f}_v}{\sum_{v\in A_u}  r_{uv}},
 \end{equation}
\end{small}
\vspace{-0.1cm}
\noindent where $\mathbf{p}_u$ is the average filter embedding of local first order neighbors of user $u$ given the graph $G_u$.

The above pooling technique is adopted for summarizing first order user-centric network, i.e., the direct connected neighbors of user $u$. For modeling the up to $L$-th higher order user-centric network,  we extend Eq.\eqref{eq:s_first_nei_mean} to aggregate the up to $L$-th order ego-centric graph structure as:

\vspace{-0.3cm}
\begin{small}
 \begin{flalign}
 \mathbf{h}^1_i&=\frac{\sum_{j\in A_{i}} (a_{ij} \mathbf{f}_j)}{\sum_{j\in A_{i}}  a_{ij}}, \quad
\forall l\geq 2, \mathbf{h}^l_i=\frac{\sum_{j\in A_{i}} (a_{ij}\mathbf{h}^{l-1}_j)}{\sum_{j\in A_{i}}  a_{ij}} \label{eq:s_k_nei_con} \\
 \mathbf{p}_u&=\frac{1}{L}\sum_{l=1}^L \mathbf{h}^l_u
 \end{flalign}
\end{small},
\vspace{-0.2cm}

\noindent where $a_{ij}$ is an edge weight in edge weight matrix $\mathbf{A}$~(Eq.\eqref{eq:cfmat2graph}). $A_i$ is the subset that directly connects to node $i$ in this matrix.  Eq.\eqref{eq:s_k_nei_con} calculates each node's $l$-th order ego-centric graph representation, and Eq.\eqref{eq:s_k_nei_con} averages each layer's representation as the ego-centric representation.

However, this simple average aggregation fails, as it does not account for the different higher graph structure in the modeling process. As illustrated in Figure~\ref{tab:primilary}, as $l$ increases, each $l$-th order neighbor becomes more distant from the current user,  and the ability of distant neighbors to expose this user's sensitive information becomes smaller compared to the closer neighbors.

\textbf{Local Value Aggregation}. For each user $u$, instead of directly modeling the up-to-K-th order subgraph representation, we argue that her sensitive attribute is better not exposed by any $l$-th layer user-centric graph structure representation $\mathbf{h}^l_{u}$. Let $V^l$ denote the value function of the $l$-th subgraph structure, we have the following value function:

\vspace{-0.4cm}
\begin{small}
\begin{equation} \label{eq:s_value}
V^l_{S}= \mathbb{E}_{(u,v,r,x_u)} \sum_{k=1}^{K } x_{uk} ln \mathcal{D}^k(\mathbf{h}^l_{u}).
\end{equation}
\end{small}
\vspace{-0.2cm}

After that, the subgraph based value function in Eq.\eqref{eq:f1_loss2}
is a combination of the up to $L$-th order value function:l
\vspace{-0.2cm}
\begin{small}
\begin{equation}\label{eq:s_value_agg}
V_{S}=\lambda_1 V^1_{S}+...+\lambda_l V^l_{S}+...+\lambda_L V^L_{S}=\sum_{l=1}^L \lambda_l V^l_{S},
\end{equation}
\end{small}
where $\lambda_l$ is a balance parameter that needs to be tuned to balance different $l$-th order value function. The larger the $\lambda_l$, the more important this $l$-th order value function.

\textbf{Learning based Aggregation.} The above local value aggregation function  needs to involve human labor to manually tune the balance parameter $\lambda_l$. For each user $u$, we propose to directly learn the ego-centric representation $p_u$ with each $l$-th layer representation $\mathbf{h}^l_u$. We propose to adopt deep neural networks to learn the sub-graph representation. We use a Multilayer Perceptron (MLP) to model the non-linear aggregation of all layers for sub graph representation, as MLPs are powerful to approximate any universal complex functions ~\cite{goodfellow2016deep}. Specifically, the learning based aggregation involves an MLP to learn the final ego-centric graph embedding as: 

\vspace{-0.2cm}
\begin{small}
 \begin{equation}\label{eq:s_learning_agg}
 \mathbf{p}_u=MLP(\mathbf{h}^1_u, \mathbf{h}^2_u, ...,\mathbf{h}^L_u),
 \end{equation}
\end{small}
\vspace{-0.2cm}

\noindent where the learnable parameters are the parameters in the MLP structure, which can be learned with other parameters in a unified training procedure.

Please note that, someone may argue that there are advanced graph embedding models with carefully designed architecture for learning the ego-centric graph representation. As the focus of this paper is not to design more sophisticated graph embedding models, we use a simple yet effective summary network for ego-centric graph representation, and focus on whether modeling the graph structure is effective for fair representation learning.

\section{Theoretical Analysis}
In this section, we theoretically analyze the implications of our proposed model. 

Specifically, in the supplementary material, we show that the overall value function in Eq.\eqref{eq:original_likelihood} can be  seen as independent combinations of each sub discriminator $\mathcal{D}^k$ with attribute $k$. Without loss of generality, we consider the overall value function with regard to the $k$-th attribute is:

\vspace{-0.4cm}
\begin{flalign} \label{eq:original_dis}
V(\mathcal{F},\mathcal{D}^k)=&\mathop{\mathbb{E}}\limits_{(u,v,r,x)\sim p(\mathbf{E},\mathbf{R},\mathbf{X})}
[\ln q_{\mathcal{R}}(r|\mathcal{F}(G_u,\mathbf{E},\mathbf{X}))-\\ \nonumber
&\lambda K \ln q_{\mathcal{D}^k}(x_{uk}|\mathcal{F}(G_u,\mathbf{E},\mathbf{X}))].
\end{flalign}

\noindent

Since both the rating prediction part and the discriminator rely on the filtered embeddings $\mathbf{F}=\mathcal{F}(G,\mathbf{E},\mathbf{X})$ that is directly filtered from the original embeddings $\mathbf{E}$ from any recommendation model, we define an alternative distribution over the filtered embedding space $\mathbf{F}$ as follows:
\vspace{-0.2cm}
\begin{footnotesize}
\begin{flalign} \label{eq:filter_dis}  
 \hat{p}(\mathbf{f}_u,\mathbf{f}_v,\mathbf{p}_u,r,x) &=\int_{\mathbf{e}_u,\mathbf{e}_v} \hat{p}(\mathbf{e}_u, \mathbf{e}_v,\mathbf{f}_u,\mathbf{f}_v,\mathbf{p}_u,r,x)d(\mathbf{e}_u, \mathbf{e}_v) \nonumber\\
 &=\int_{\mathbf{e}_u,\mathbf{e}_v} p(\mathbf{e}_u, \mathbf{e}_v,r,x)p_{\mathcal{F}}(\mathbf{f}_u,\mathbf{f}_v,\mathbf{p}_u|\mathbf{e}_u, \mathbf{e}_v)d(\mathbf{e}_u, \mathbf{e}_v)   \nonumber \\
 &=\int_{\mathbf{e}_u,\mathbf{e}_v} p(\mathbf{e}_u,\mathbf{e}_v,r,x)\delta(\mathcal{F}(G_u,\mathbf{E},\mathbf{X})=(\mathbf{f}_u,\mathbf{f}_v,\mathbf{p}_u))d(\mathbf{e}_u,\mathbf{e}_v).
\end{flalign}
\end{footnotesize}
\vspace{-0.2cm}

With the alternative distribution that relies on the filtered embedding space in Eq.\eqref{eq:filter_dis}, we replace Eq.\eqref{eq:original_dis} to:

\vspace{-0.2cm}
\begin{flalign}\label{eq:filter_rewrite}
V(\mathcal{F},\mathcal{D}^k)=&\mathop{\mathbb{E}}\limits_{(\mathbf{f}_u,\mathbf{f}_v,\mathbf{p}_u,r,x)\sim \hat{p}(\mathbf{f}_u,\mathbf{f}_v,\mathbf{p}_u,r,x)}[\ln q_{\mathcal{R}}(r|\mathcal{F}(G_u,\mathbf{E},\mathbf{X}))
-\\ \nonumber
&\lambda K \ln q_{\mathcal{D}^k} (x_{uk}|\mathcal{F}(G_u,\mathbf{E},\mathbf{X}))]. 
\end{flalign}
\vspace{-0.2cm}
After that, we have the following propositions.

\begin{lemma} If the discriminator network has enough capacity, the optimal solution of $q^{*}_{\mathcal{D}^k}$ is $\hat{p}(x_{uk}|\mathbf{f}_u,\mathbf{p}_u)$.
\end{lemma}

\begin{proof}  
Given the equality constraints of the predicted probability distribution $\sum_{x}q_{\mathcal{D}^k}(x_{uk}|(\mathbf{f}_u, \mathbf{p}_u))=1$, we can obtain the Lagrangian dual optimization problem, and solve it. We show the details of this proof in the supplementary material.
\end{proof}

\begin{lemma}
When $\lambda\rightarrow \infty$, if both the filter $\mathcal{F}$ and the discriminator $\mathcal{D}$ have enough capacity, and at each step the discriminator and filter are allowed to reach their optimal values. Then, the filtered embedding space is conditionally independent with the sensitive attributes.
\end{lemma}

\begin{proof}
In fact, when $\lambda\rightarrow \infty$, $V_R$ disappears in Eq.\eqref{eq:original_likelihood}. And the above proposition could be easily validated if we check the proofs in Section 4 of the original GAN model ~\cite{goodfellow2014generative}.
\end{proof}

However, the above proposition is too strict as when $\lambda\rightarrow \infty$, the rating prediction objective $V_R$ is discarded in the proposed model. In the following, we do not give restrictions on $\lambda$, and give a detailed analysis of the objective function in Eq.\eqref{eq:original_likelihood}.

\begin{theorem}
Given enough capacity of the discriminator network, the objective function in Eq.\eqref{eq:filter_rewrite} is equivalent to  $\min_{\mathcal{F}} H(\mathbf{R}|\mathbf{F})-\lambda K H(\mathbf{X}_k|(\mathbf{F},\mathbf{P}))$, i.e.,
minimizing the conditional entropy between the ratings and filtered embeddings, while maximizing the conditional entropy between the sensitive attribute and the filtered embeddings.
\end{theorem}

\begin{proof}
By replacing the best discriminator in Lemma 1,  the objective goal in Eq. \eqref{eq:filter_rewrite} is equal to:
\vspace{-0.2cm}
\begin{flalign}
\arg\min_{\mathcal{F}}V(\mathcal{F},\mathcal{D}^{k^*}) &=\mathop{\mathbb{E}}\limits_{(\mathbf{f}_u,\mathbf{f}_v,\mathbf{p}_u,r,x)\sim\hat{p}(\mathbf{f}_u,\mathbf{f}_v,\mathbf{p}_u,r,x)} [\ln q_R (r|(\mathbf{f}_u,\mathbf{f}_v))
-\\ \nonumber
&\quad \lambda K \ln \hat{p} (x_{uk}|(\mathbf{f}_u,\mathbf{p}_u))] \nonumber  \\ 
&= -H(\mathbf{R}|\mathbf{F})+\lambda K H(\mathbf{X}_k|(\mathbf{F},\mathbf{P})). \nonumber \\ 
\end{flalign}
\end{proof}
\vspace{-0.2cm}

By combining Eq.\eqref{eq:original_likelihood} and Theorem 3, we can easily extend the above theory to multiple sensitive attributes as: 
\vspace{-0.2cm}
\begin{flalign}
\arg\min_{\mathcal{F}}V(\mathcal{F},\mathcal{D}) = - H(\mathbf{R}|\mathbf{F})+\lambda\sum_{k=1}^{K} H(\mathbf{X}_k|(\mathbf{F},\mathbf{P})).
\end{flalign}
\vspace{-0.2cm}

Therefore, we have the following theorem as:

\begin{theorem}
Given enough capacity of the discriminator network, the objective function in Eq.\eqref{eq:original_likelihood} is equivalent to  $\min_{\mathcal{F}} [-H(\mathbf{R}|\mathbf{F})+\lambda\sum_{k=1}^K H(\mathbf{X}_k|(\mathbf{F},\mathbf{P}))]$, i.e., minimizing the conditional entropy between the ratings and filtered embeddings, while maximizing the conditional entropy between each sensitive attribute and the filtered embeddings.
\end{theorem}

\section{Experiments}

\subsection{Experimental Setup}
 
\textbf{Datasets.}  \emph{MovieLens-1M}  is a benchmark dataset for recommender systems~\cite{harper2015movielens}. The dataset contains 6040 users' 1 million rating records to about 4000 movies. Users are associated with three attributes, including gender~(two classes), age~(seven classes), and occupation~(21 classes)~\footnote{https://grouplens.org/datasets/movielens/}. Similar as the previous works for fairness based recommendation~\cite{bose2019compositional}, we split the historical ratings into training and test with a ratio of 9:1.

\emph{Lastfm-360K} is a music recommendation dataset that contains users' ratings to artists collected from the music website of \emph{Last.fm}~\cite{celma2009music}. The dataset contains about 360 thousand users' 17 million records to 290 thousand artists. We treat the play times as the rating values. As the detailed rating values are in a large range, we first preprocess ratings with log transformations, and then normalize ratings into range 1 to 5.  Users are associated with a profile, including gender~(two classes), and age. For the age attribute, we transform ages into three classes~\footnote{http://ocelma.net/MusicRecommendationDataset/lastfm-360K.html}. Similar as many classical recommendation data split approaches, we split the historical ratings into training, validation, and test parts with the ratio of 7:1:2.

\begin{table*}[htb]
\caption{Performance on MovieLens-1M. We test performance on both the single attribute and the compositional setting with multiple sensitive attributes~(denoted as Com.). Smaller values mean better performance.} \label{tab:ml_all}
\vspace{-0.3cm}
\scalebox{0.80}{
\begin{tabular}{|l|l|l|l|l|l|l|l|l|l|l|l|l|}
\hline
\multirow{2}{*}{Sensitive Att.} & \multicolumn{2}{l|}{PMF} & \multicolumn{2}{l|}{GCN} & \multicolumn{2}{l|}{Non-parity} & \multicolumn{2}{l|}{ICML\_2019} & \multicolumn{2}{l|}{FairGo\_PMF} & \multicolumn{2}{l|}{FairGo\_GCN} \\ \cline{2-13} 
 & RMSE & AUC/F1 & RMSE & AUC/F1 & RMSE & AUC/F1 & RMSE & AUC/F1 & RMSE & AUC/F1 & RMSE & AUC/F1 \\ \hline
Gen. & \multirow{6}{*}{0.8681} & 0.6615 & \multirow{6}{*}{0.8564} & 0.7041 & 0.8621 & 0.8428 & 0.9203 & 0.5175 & 0.9150 & 0.5042 & 0.9068 & 0.5042 \\ \cline{1-1} \cline{3-3} \cline{5-13} 
Age &  & 0.3821 &  & 0.4215 & \textbackslash{} & \textbackslash{} & 0.9203 & 0.3420 & 0.9059 & 0.3220 & 0.9051 & 0.3140 \\ \cline{1-1} \cline{3-3} \cline{5-13} 
Occ. &  & 0.1332 &  & 0.1485 & \textbackslash{} & \textbackslash{} & 0.9186 & 0.1190 & 0.9367 & 0.1130 & 0.9069 & 0.1070 \\ \cline{1-1} \cline{3-3} \cline{5-13} 
Com.-Gen. &  & 0.6615 &  & 0.7041 & \textbackslash{} & \textbackslash{} & \multirow{3}{*}{0.9191} & 0.5389 & \multirow{3}{*}{0.9325} & 0.5026 & \multirow{3}{*}{0.9185} & 0.5134 \\ \cline{1-1} \cline{3-3} \cline{5-7} \cline{9-9} \cline{11-11} \cline{13-13} 
Com.-Age &  & 0.3821 &  & 0.4215 & \textbackslash{} & \textbackslash{} &  & 0.3620 &  & 0.3380 &  & 0.3260 \\ \cline{1-1} \cline{3-3} \cline{5-7} \cline{9-9} \cline{11-11} \cline{13-13} 
Com.-Occ. &  & 0.1332 &  & 0.1485 & \textbackslash{} & \textbackslash{} &  & 0.1240 &  & 0.1060 &  & 0.1250 \\ \hline
\end{tabular}
}
\end{table*}

\begin{table*}[htb] 
\caption{Performance on Lastfm-360K.} \label{tab:lastfm_all}
\vspace{-0.3cm}
\scalebox{0.80}{
\begin{tabular}{|l|l|l|l|l|l|l|l|l|l|l|l|l|}
\hline
\multirow{2}{*}{Sensitive Att.} & \multicolumn{2}{l|}{PMF} & \multicolumn{2}{l|}{GCN} & \multicolumn{2}{l|}{Non-parity} & \multicolumn{2}{l|}{ICML\_2019} & \multicolumn{2}{l|}{FairGo\_PMF} & \multicolumn{2}{l|}{FairGo\_GCN} \\ \cline{2-13} 
 & RMSE & AUC/F1 & RMSE & AUC/F1 & RMSE & AUC/F1 & RMSE & AUC/F1 & RMSE & AUC/F1 & RMSE & AUC/F1 \\ \hline
Gen. & \multirow{4}{*}{0.7112} & 0.5506 & \multirow{4}{*}{0.7034} & 0.5696 & 0.7346 & 0.6649 & 0.7259 & 0.5409 & 0.7096 & 0.5428 & 0.7072 & 0.5354 \\ \cline{1-1} \cline{3-3} \cline{5-13} 
Age &  & 0.4695 &  & 0.4716 & \textbackslash{} & \textbackslash{} & 0.7204 & 0.4682 & 0.7195 & 0.4689 & 0.7061 & 0.4672 \\ \cline{1-1} \cline{3-3} \cline{5-13} 
Com.-Gen &  & 0.5506 &  & 0.5696 & \textbackslash{} & \textbackslash{} & \multirow{2}{*}{0.7173} & 0.5379 & \multirow{2}{*}{0.7081} & 0.5347 & \multirow{2}{*}{0.7049} & 0.5367 \\ \cline{1-1} \cline{3-3} \cline{5-7} \cline{9-9} \cline{11-11} \cline{13-13} 
Com.-Age &  & 0.4695 &  & 0.4716 & \textbackslash{} & \textbackslash{} &  & 0.4688 &  & 0.4681 &  & 0.4662 \\ \hline
\end{tabular}
}
\end{table*}

\textbf{Experimental Setup and Evaluation.} Our model involves three steps, a pretrained recommendation algorithm, followed by the proposed \shortname model for fairness consideration. After that, we need to evaluate the fairness performance.
We first use the training data to complete the first two steps, with the rating records in the training data as ground truth preference data, and the user attributes in the training data as ground truth sensitive information. The validation data is used for model parameter tuning.   When finishing the model training, in order to evaluate whether the sensitive information is exposed by the learned model, similar as many works for fairness models~\cite{icml2018learning,bose2019compositional}, we randomly select 80\% users' attributes as ground truth and train a linear classification model by taking the learned fair representations. We test the classification accuracy on the remaining 20\% users for fairness evaluation.

For measuring the recommendation performance, we use Root Mean Squared Error~(RMSE) metric~\cite{koren2009MF}. For measuring the fairness goal, we calculate classification performance of the 20\% test users. As the binary attribute~(i.e., gender) is imbalanced on both datasets, with about 70\%  males and 30\% females, we use Area Under Curve~(AUC) metric for measuring binary classification performance. For the remaining attributes with multiple values, we use micro-averaged F1 measure~\cite{DBLPHandC18}. AUC or F1 can be used as a measure of whether the sensitive gender information is exposed in the representation learning process. The smaller values of these classification metrics denote better fairness performance with less sensitive information leakage.

As our proposed model is model-agnostic and can be applied to fair recommendation with multiple attributes, we design several experiments with different settings for model evaluation. First, we choose two base recommendation models: a free latent embedding model of PMF~\cite{mnih2008PMF} and a state-of-the-art GCN based recommendation model~\cite{AAAI2020revi}. As this GCN based recommendation model is originally designed with ranking based loss function, we modify it to the rating based loss function, and add the detailed rating values in the graph convolution process to facilitate our setting. Second, as the sensitive attribute setting varies, we perform experiments on both the single sensitive attribute setting and the compositional setting with multiple sensitive attributes. For example, we have three settings of one single sensitive attribute, i.e., gender, age, and occupation, and one compositional setting of three sensitive attributes on MovieLens-1M dataset.

\textbf{Baselines and Parameter Setting.}
We compare our proposed model with the following baselines, including state-of-the-art recommendation models of PMF~\cite{mnih2008PMF} and GCN~\cite{AAAI2020revi}.
To explicitly model the fairness metrics, we choose a state-of-the-art model that can leverage multiple sensitive attributes, i.e., ICML\_2019~\cite{bose2019compositional} as a baseline. Besides, we choose a fairness regularization based  model, i.e., Non-parity~\cite{wei2017beyond} as a baseline. Given a binary valued sensitive attribute, Non-parity defines different metrics for unfairness and incorporates the corresponding unfairness regularization  terms for recommendation. Each unfairness metric is based on the average rating prediction of the advantaged group with attribute value of 1 and the remaining group with attribute value of 0. Due to the constraint that Non-parity is suitable for binary valued attributes, we apply this baseline to the gender attribute on the two datasets.

In practice, in our proposed \shortname model, we choose MLPs as the detailed architecture of each filter and each discriminator.  The filtered embedding size is set as $D=64$. For MovieLens dataset, each filter network has 3 layers which the hidden layer sizes as 128 and 64 respectively, and the discriminator has 4 layers which the hidden layer sizes are 16 and 8 respectively. For Lastfm-360K dataset, each filter network has 4 layers with the hidden layer sizes as 128, 64, 32 respectively, and  each discriminator has 4 layers with the hidden layer sizes as 16, 8, and 4. We use LeaklyReLU as the activation function. The balance parameter $\lambda$ in Eq.\eqref{eq:original_likelihood} is set as 0.1 on MovieLens and 0.2 on Lastfm-360K. All the parameters are differentiable in the objective function, and we use Adam optimizer with the initial learning rate of 0.005.

\begin{table*}[htb]
\caption{Performance of different summary networks for ego-centric structure on MovieLens-1M, with ``value'' denotes the local value function aggregation, and ``learning'' denotes the learning based aggregation.
}\label{tab:ml_2layer} 
\vspace{-0.3cm}
\scalebox{0.85}{ 
\begin{tabular}{|l|l|l|l|l|l|l|l|l|l|l|l|l|}
\hline
\multirow{3}{*}{Senstive Att.} & \multicolumn{6}{c|}{FairGo\_PMF} & \multicolumn{6}{c|}{FairGo\_GCN} \\ \cline{2-13} 
 & \multicolumn{2}{c|}{L=1} & \multicolumn{2}{c|}{L=2(value)} & \multicolumn{2}{c|}{L=2(learning)} & \multicolumn{2}{c|}{L=1} & \multicolumn{2}{c|}{L=2(value)} & \multicolumn{2}{c|}{L=2(learning)} \\ \cline{2-13} 
 & RMSE & AUC/F1 & RMSE & AUC/F1 & RMSE & AUC/F1 & RMSE & AUC/F1 & RMSE & AUC/F1 & RMSE & AUC/F1 \\ \hline
Gen. & 0.9150 & 0.5042 & 0.9082 & 0.5045 & 0.9055 & 0.5018 & 0.9068 & 0.5042 & 0.9070 & 0.5065 & 0.9004 & 0.5014 \\ \hline
Age & 0.9059 & 0.3220 & 0.9077 & 0.3200 & 0.9001 & 0.3160 & 0.9051 & 0.3140 & 0.9036 & 0.3080 & 0.9045 & 0.3100 \\ \hline
OCC. & 0.9367 & 0.1130 & 0.9332 & 0.1150 & 0.9186 & 0.1060 & 0.9069 & 0.1070 & 0.9079 & 0.1090 & 0.9010 & 0.1000 \\ \hline
\end{tabular} 
}
\end{table*}

\begin{table*}[htb]
\vspace{-0.3cm}
\caption{Performance of different summary networks for ego-centric structure on Lastfm-360K.
}\label{tab:lastfm_2layer} 
\vspace{-0.3cm}
\scalebox{0.85}{ 
\begin{tabular}{|l|l|l|l|l|l|l|l|l|l|l|l|l|}
\hline
\multirow{3}{*}{Senstive Att.} & \multicolumn{6}{c|}{FairGo\_PMF} & \multicolumn{6}{c|}{FairGo\_GCN} \\ \cline{2-13} 
 & \multicolumn{2}{c|}{L=1} & \multicolumn{2}{c|}{L=2(value)} & \multicolumn{2}{c|}{L=2(learning)} & \multicolumn{2}{c|}{L=1} & \multicolumn{2}{c|}{L=2(value)} & \multicolumn{2}{c|}{L=2(learning)} \\ \cline{2-13} 
 & RMSE & AUC/F1 & RMSE & AUC/F1 & RMSE & AUC/F1 & RMSE & AUC/F1 & RMSE & AUC/F1 & RMSE & AUC/F1 \\ \hline
Gen. & 0.7096 & 0.5428 & 0.7025 & 0.5361 & 0.7020 & 0.5357 & 0.7072 & 0.5354 & 0.7091 & 0.5442 & 0.7068 & 0.5337 \\ \hline
Age & 0.7195 & 0.4689 & 0.7082 & 0.4678 & 0.7099 & 0.4666 & 0.7061 & 0.4672 & 0.7015 & 0.4691 & 0.7047 & 0.4669 \\ \hline 
\end{tabular} 
}
\vspace{-0.1cm}
\end{table*}

\subsection{The Overall Performance}

We report the overall results in Table~\ref{tab:ml_all} and Table~\ref{tab:lastfm_all}. In these two tables, our proposed \shortname adopts the simple ego-centric graph representation with weighted first order aggregation in Eq.\eqref{eq:s_first_nei_mean}. We have several observations from this table. First, when comparing the results of two state-of-the-art recommendation models of PMF and GCN, GCN has better recommendation performance~(smaller RMSE values) and exposes more sensitive information~(larger classification metric values). This is due to the fact that GCN directly models the graph structure for embedding learning, which alleviates the sparsity issue, and discovers some hidden features that are correlated with sensitive feature set. Second, we observe that all models that directly consider the sensitive information filter would decrease the recommendation performance to 5\% to 10\%, as we need to eliminate any latent dimensions that are useful for rating, but may expose the sensitive attribute. 
Non-parity does not achieve satisfactory performance on these two dataset. We guess a possible reason is that, the Non-parity baseline measures the discrepancy of the predicted ratings of the two groups, and does not directly remove sensitive attribute information in embeddings.
When comparing these fairness-aware models, FairGo\_GCN considers the correlation of entities from a graph perspective and reaches the best performance for both the rating prediction and fairness elimination task. As to the FairGo\_PMF, it has better fairness performance compared to ICML\_2019, but the recommendation performance of RMSE is not consistent, as it shows worse performance in the compositional setting. This is due to the fact that the base model~(i.e., PMF) in FairGo\_PMF does not perform well as  base graph embedding model. Please note that, the fairness results of occupation have a large variance. We guess a possible reason is that, the occupation values are imbalanced and have 21 distinct values. Given limited 6040 users, the adversary network is hard to train in practice. For the Lastfm dataset, Table~\ref{tab:lastfm_all} shows a similar overall trend  as analyzed above. Therefore, we conclude that our proposed \shortname framework could improve fairness with very little recommendation accuracy loss. By using a more advanced base recommendation model, our proposed FairGo\_GCN reaches the best performance for both recommendation and fairness. In the following, we choose FairGo\_GCN for detailed analysis since FairGo\_GCN shows better recommendation and fairness compared to FairGo\_PMF.

\subsection{Detailed Model Analysis}

\textbf{Performance of different user-centric subgraph modeling.} 
In this part, we would like to explore the performance under different higher order graph modeling techniques. We focus on the experimental settings on single attribute. We conduct experiments on the two proposed approaches:  local value function aggregation~(Eq.\eqref{eq:s_value_agg}) and learning based aggregation~(Eq.\eqref{eq:s_learning_agg}) with second-order user-centric subgraph. Specifically, the local value function aggregation is calculated as~$V_{S}=\lambda_1 V^1_{S}+\lambda_2 V^2_{S}$, while the two parameters~$\lambda_1$ and~$\lambda_2$ are set to~$4:1$ in FairGo\_PMF and~$1:1$ in FairGo\_GCN. The MLP structure in the learning based aggregation has two non-linear layers and one linear layer. The results on MovieLens and Lastfm-360K are shown in Table~\ref{tab:ml_2layer} and Table~\ref{tab:lastfm_2layer}. As can be observed from both tables, the learning based aggregation shows the best performance for all settings. The local value function aggregation shows better performance than the first order neighborhood modeling for most settings, as it relies on manual tuning of balance parameters. Therefore, we empirically conclude higher order graph structure can achieve better fairness results. By using the learning based subgraph modeling, our proposed model can further improve recommendation accuracy and fairness. However, we notice that modeling the higher order graph structure also introduces more runtime, and more difficulty in the model training process.

Please note that when considering the second order local graph structure, on average each user's ego-centric graph includes 10\% nodes on MovieLens-1M and about 5000 nodes on Lastfm-360K. If we further increase the layer size to 3, each user's subgraph largely overlaps with the subgraph of other users.  
Therefore, we do not report the results with more than 3 layers.

 \vspace{-0.4cm}
\begin{figure} [htb]
	\begin{center}
		\includegraphics[width=80mm]{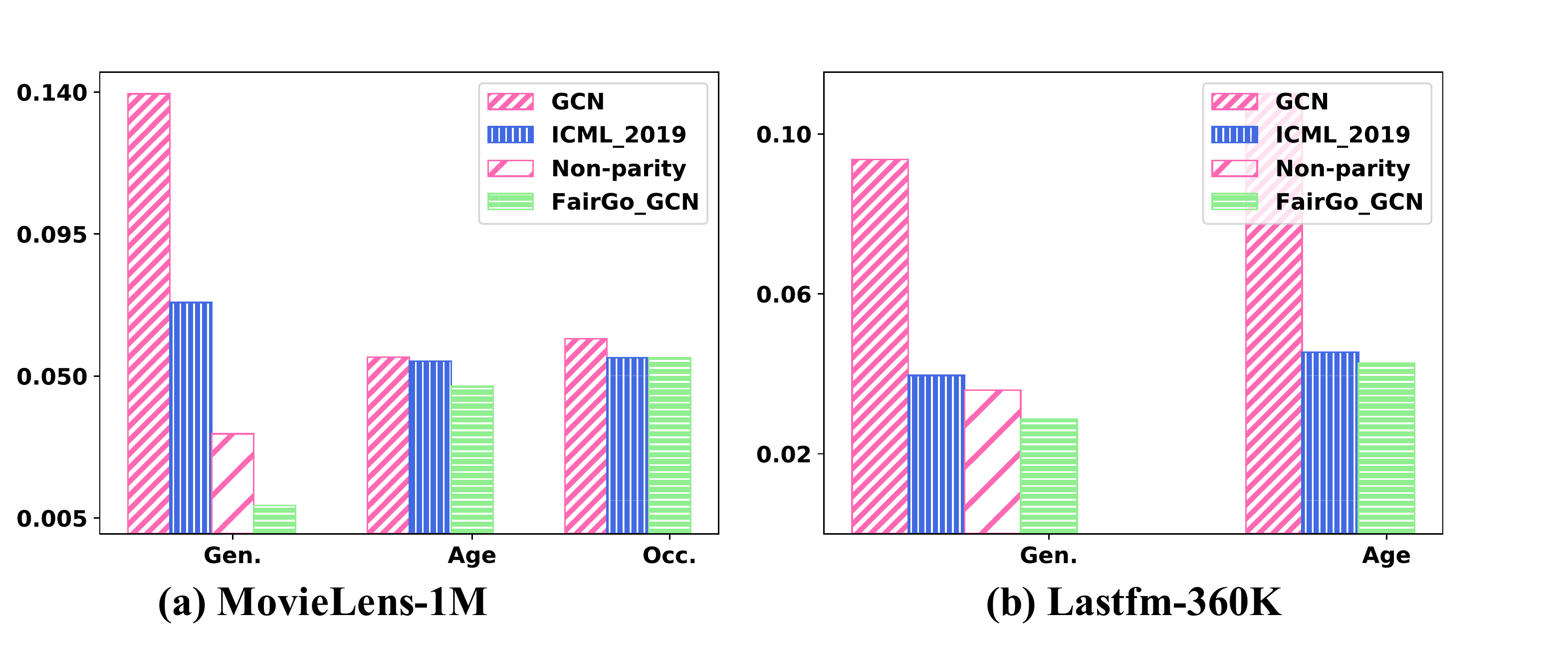}
	\end{center}
	\vspace{-0.4cm}
	\caption{\small{Performance of statistical parity measure.}}\label{fig:group_fair1}
\end{figure}

\vspace{-0.5cm}
\begin{figure} 
	\begin{center}
		\includegraphics[width=80mm]{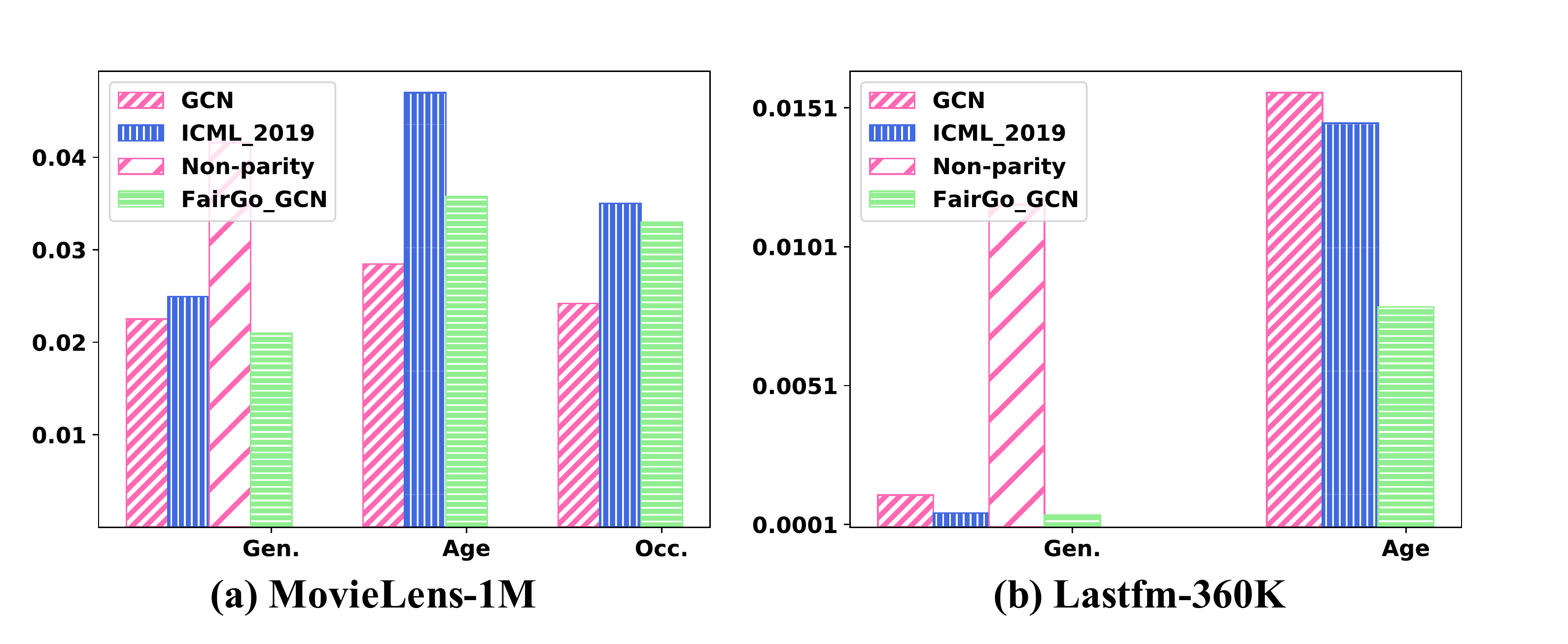}
	\end{center}
	\vspace{-0.4cm}
	\caption{\small{Performance of equal opportunity measure. }}\label{fig:group_fair2}
\end{figure}

\textbf{Relation to group fairness.}
As there are many fairness metrics, in this part, we would show the results of  our proposed model on group fairness measures. For all group fairness based metrics, statistical parity and equal opportunity are widely used. For attributes with multiple values, we borrow the idea of statistical parity. For attributes with binary values, we use the equal opportunity. The concrete formulas of two group fairness metrics is recorded in the supplementary material.  We show the results of statistical parity and equal opportunity in Figure~\ref{fig:group_fair1} and Figure~\ref{fig:group_fair2}. In short, our proposed model achieves the best results for binary gender attribute. Our proposed model also reaches the best results for the two group based fairness metrics on Lastfm-360K dataset.

\section{Conclusions and Future Work}
In this paper, we argued that most current works on fairness based models assumed independence of instances, and could not be well applied to the recommendation scenario. To this end, we proposed a \shortname model that considered fairness
from a graph perspective for any current recommendation models.  The proposed framework is model-agnostic and can be applied to multiple sensitive attributes. Experimental results on real-world datasets clearly showed the effectiveness of our proposed model.  In the future,  we would like to explore the potential of our proposed model to domain specific applications, such as job or education recommendation.

\section*{Acknowledgements}

This work was supported in part by grants from the National Natural Science Foundation of China~(Grant No. 61972125, U19A2079, U1936219, 61932009, 91846201),
and CAAI-Huawei MindSpore Open Fund.

\newpage
\bibliographystyle{ACM-Reference-Format}
\bibliography{2020-fair}
\newpage
\appendix

\section{Proofs}
We give the details of some proofs in Section 4, i.e., correlation between the overall value function~(Eq.(4)) and the sub value function~(Eq.(16)), and proofs of lemma 1.\\

\subsection{Correlation between the overall value function~(Eq.(4)) with multiple attributes and the sub value function~(Eq.(16)) that deals with a single attribute.}

The overall value function can be written as:

\begin{footnotesize}
\begin{flalign} \label{eq:original_overall}
V(\mathcal{F},\mathcal{D})&=\mathop{\mathbb{E}}\limits_{(u,v,r,x)\sim p(\mathbf{E},\mathbf{R},\mathbf{X})}
[\ln q_{\mathcal{R}}(r|(\mathbf{f}_u,\mathbf{f}_v,\mathbf{p}_u))
-\lambda \ln q_{\mathcal{D}}(x_{u}|(\mathbf{f}_u,\mathbf{f}_v,\mathbf{p}_u))]. \nonumber \\
&=\mathop{\mathbb{E}}\limits_{(u,v,r,x)\sim p(\mathbf{E},\mathbf{R},\mathbf{X})}
[\ln q_{\mathcal{R}}(r|(\mathbf{f}_u,\mathbf{f}_v,\mathbf{p}_u))]
\nonumber \\
&\quad -\lambda\sum_{k=1}^{K}\mathop{\mathbb{E}}\limits_{(u,v,r,x)\sim p(\mathbf{E},\mathbf{R},\mathbf{X})} \ln q_{\mathcal{D}^{k}}(x_{uk}|(\mathbf{f}_u,\mathbf{f}_v,\mathbf{p}_u)) \nonumber\\
&=1/K\sum_{k=1}^{K}\mathop{\mathbb{E}}\limits_{(u,v,r,x)\sim p(\mathbf{E},\mathbf{R},\mathbf{X})}
[\ln q_{\mathcal{R}}(r|(\mathbf{f}_u,\mathbf{f}_v,\mathbf{p}_u))\nonumber \\
&\quad -\lambda K \ln q_{\mathcal{D}^k}(x_{uk}|(\mathbf{f}_u,\mathbf{f}_v,\mathbf{p}_u))],\nonumber \\ \tag{22}
\end{flalign}
\end{footnotesize}
\noindent
where $(\mathbf{f}_u,\mathbf{f}_v,\mathbf{p}_u)=\mathcal{F}(G_u,\mathbf{E},\mathbf{X})$ is the mapping function from the origin embedding space to the filtered embedding space, and $p_u$ is summarized from the filtered embedding space. This Eq.~\eqref{eq:original_overall} corresponds to Eq.(4) in Section 3. Thus, the overall value function can be easily seen as a combination of each sub discriminator $\mathcal{D}^k$ with attribute $k$. Without loss of generality, we consider the overall value function with regard to the $k$-th attribute as:

\begin{footnotesize}
\begin{equation} \label{eq:original_dis}
V(\mathcal{F},\mathcal{D}^k)=\mathop{\mathbb{E}}\limits_{(u,v,r,x)\sim p(\mathbf{E},\mathbf{R},\mathbf{X})}
[\ln q_{\mathcal{R}}(r|(\mathbf{f}_u,\mathbf{f}_v,\mathbf{p}_u))-\lambda K \ln q_{\mathcal{D}^k}(x_{uk}|(\mathbf{f}_u,\mathbf{f}_v,\mathbf{p}_u))].\tag{23}
\end{equation}
\end{footnotesize}
\noindent This Eq.~\eqref{eq:original_dis} corresponds to Eq.(16) in Section 4.

Since both the rating prediction part and the discriminator rely on the filtered embeddings $\mathbf{F}=\mathcal{F}(G_u,\mathbf{E},\mathbf{X})$, we define an alternative distribution over the filtered embedding space $\mathbf{F}$ as follows:
\begin{footnotesize}
\begin{flalign} \label{eq:filter_dis}
 \hat{p}(\mathbf{f}_u,\mathbf{f}_v,\mathbf{p}_u,r,x) &=\int_{\mathbf{e}_u,\mathbf{e}_v} \hat{p}(\mathbf{e}_u, \mathbf{e}_v,\mathbf{f}_u,\mathbf{f}_v,\mathbf{p}_u,r,x)d(\mathbf{e}_u, \mathbf{e}_v) \nonumber\\
 &=\int_{\mathbf{e}_u,\mathbf{e}_v} p(\mathbf{e}_u, \mathbf{e}_v,r,x)p_{\mathcal{F}}(\mathbf{f}_u,\mathbf{f}_v,\mathbf{p}_u|\mathbf{e}_u, \mathbf{e}_v)d(\mathbf{e}_u, \mathbf{e}_v)   \nonumber \\
 &=\int_{\mathbf{e}_u,\mathbf{e}_v} p(\mathbf{e}_u,\mathbf{e}_v,r,x)\delta(\mathcal{F}(G_u,\mathbf{E},\mathbf{X})=(\mathbf{f}_u,\mathbf{f}_v,\mathbf{p}_u))d(\mathbf{e}_u,\mathbf{e}_v).\tag{24}
\end{flalign}
\end{footnotesize}
With the alternative distribution that relies on the filtered embedding space in Eq.\eqref{eq:filter_dis}, we replace Eq.\eqref{eq:original_dis} to:

\vspace{-0.2cm}
\begin{footnotesize}
\begin{flalign}\label{eq:filter_rewrite1}
V(\mathcal{F},\mathcal{D}^k)&=\mathop{\mathbb{E}}\limits_{(\mathbf{f}_u,\mathbf{f}_v,\mathbf{p}_u,r,x)\sim \hat{p}(\mathbf{f}_u,\mathbf{f}_v,\mathbf{p}_u,r,x)}[\ln q_{\mathcal{R}}(r|(\mathbf{f}_u,\mathbf{f}_v,\mathbf{p}_u))
\nonumber \\
&-\lambda K \ln q_{\mathcal{D}^k} (x_{uk}|(\mathbf{f}_u,\mathbf{f}_v,\mathbf{p}_u))], \nonumber \\
&=\mathop{\mathbb{E}}\limits_{(\mathbf{f}_u,\mathbf{f}_v,\mathbf{p}_u,r,x)\sim \hat{p}(\mathbf{f}_u,\mathbf{f}_v,\mathbf{p}_u,r,x)}[\ln q_{\mathcal{R}}(r|\mathcal{F}(G_u,\mathbf{E},\mathbf{X})) \nonumber\\
&-\lambda K \ln q_{\mathcal{D}^k} (x_{uk}|\mathcal{F}(G_u,\mathbf{E},\mathbf{X}))].\tag{25} 
\end{flalign}
\end{footnotesize}
\noindent This Eq.~\eqref{eq:filter_rewrite1} corresponds to Eq.(18) in Section 4.

From the above, we can split multiple attributes into independent combinations of single attributes for analysis. Thus, the analysis of a single attribute can be easily extended to multiple attributes naturally.

\subsection{Proofs of Lemma 1}

\begin{lemma} If the discriminator network has enough capacity, the optimal solution of $q^{*}_{\mathcal{D}^k}$ is $\hat{p}(x_{uk}|\mathbf{f}_u,\mathbf{p}_u)$.
\end{lemma}
\begin{proof}  
We begin with the value function with regard to the $k$-th attribute in the filtered embedding space:
\begin{footnotesize}
\begin{flalign}\label{eq:filter_rewrite}
V(\mathcal{F},\mathcal{D}^k)=\mathop{\mathbb{E}}\limits_{(\mathbf{f}_u,\mathbf{f}_v,\mathbf{p}_u,r,x)\sim \hat{p}(\mathbf{f}_u,\mathbf{f}_v,\mathbf{p}_u,r,x)}&[\ln q_{\mathcal{R}}(r|\mathcal{F}(G_u,\mathbf{E},\mathbf{X}))
\nonumber \\
&-\lambda K \ln q_{\mathcal{D}^k} (x_{uk}|\mathcal{F}(G_u,\mathbf{E},\mathbf{X}))]. \tag{26}
\end{flalign}
\end{footnotesize}

Note that, $\mathbf{p}_u$ is an aggregation of $\mathbf{f}_u$ and $\mathbf{f}_v$, and $\mathbf{f}_v$ is irrelevant to the best solution for discriminator.
In the above value function, with the fixed embeddings $\mathbf{F}$, only the second term  \small{$-\lambda K \ln q_{\mathcal{D}^k} (x_{uk}|\mathcal{F}(G_u,\mathbf{E},\mathbf{X})$} is correlated with the discriminator.
Given the equality constraints of the predicted probability distribution $\sum_{x}q_{\mathcal{D}^k}(x_{uk}|(\mathbf{f}_u, \mathbf{p}_u))=1$, 
we can obtain the Lagrangian dual optimization problem:

\begin{footnotesize}
\begin{flalign}
L(\alpha(h)) = \sum_{h}^{}\alpha(h)&(1-\sum_{x}^{}q_{\mathcal{D}^k}(x_{uk}|(\mathbf{f}_u, \mathbf{p}_u)))\nonumber \\
&-\mathop{\mathbb{E}}\limits_{(\mathbf{f}_u,\mathbf{f}_v,\mathbf{p}_u,r,x)\sim\hat{p}(\mathbf{f}_u,\mathbf{f}_v,\mathbf{p}_u,r,x)}\lambda  K\ln q_{\mathcal{D}^k}(x_{uk}|(\mathbf{f}_u, \mathbf{p}_u)).  \tag{27}
\end{flalign}
\end{footnotesize}
\noindent

To seek  the maximum value of $ L(\alpha(h))$, we take the partial derivative of $q_{\mathcal{D}^k}$ and let the partial derivative equals 0. 

\begin{footnotesize}
\begin{flalign} \label{eq:filter_dis1}
&\frac{\partial L(\alpha(h))}{\partial q_{\mathcal{D}^k}^{*}(x_{uk}|(\mathbf{f}_u, \mathbf{p}_u))} \nonumber \\&= -\sum_{h}\alpha(h)-\frac{\mathbb{E}_{(\mathbf{f}_u,\mathbf{f}_v,\mathbf{p}_u,r,x)\sim\hat{p}(\mathbf{f}_u,\mathbf{f}_v,\mathbf{p}_u,r,x)}\lambda K \ln q_{\mathcal{D}^k}(x_{uk}|(\mathbf{f}_u, \mathbf{p}_u))}{\partial q_{\mathcal{D}^k}^{*}(x_{uk}|(\mathbf{f}_u, \mathbf{p}_u))} \nonumber \\ 
&=-\sum_{h}\alpha(h)-\mathbb{E}_{(\mathbf{f}_u,\mathbf{f}_v,\mathbf{p}_u,r,x)\sim\hat{p}(\mathbf{f}_u,\mathbf{f}_v,\mathbf{p}_u,r,x)}(\frac{\lambda K}{q_{\mathcal{D}^k}^{*}(x_{uk}|(\mathbf{f}_u, \mathbf{p}_u))}). \nonumber \\ \tag{28}
\end{flalign}
\end{footnotesize}

By letting Eq.\eqref{eq:filter_dis1} equals zero, and we employ the equality constraint as $\sum_{x}q_{\mathcal{D}^k}(x_{uk}|(\mathbf{f}_u, \mathbf{p}_u))=1$,  we get:

\begin{flalign} \label{eq:filter_dis2}
\sum_{h}\alpha(h) = -\lambda  K\frac{\sum_{r}\hat{p}(\mathbf{f}_u,\mathbf{f}_v,\mathbf{p}_u,r,x)}{q_{\mathcal{D}^k}^{*}(x_{uk}|(\mathbf{f}_u, \mathbf{p}_u))} 
    = -\lambda K \hat{p}(\mathbf{f}_u,\mathbf{p}_u).\tag{29}
\end{flalign}

After that, we substitute Eq.\eqref{eq:filter_dis2} back to Eq.\eqref{eq:filter_dis1} to get the optimal discriminator as:

\begin{equation}\label{eq:best_dis}
q_{\mathcal{D}^k}^{*}(x_{uk}|(\mathbf{f}_u,\mathbf{p}_u))=\hat{p}(x_{uk}|\mathbf{f}_u,\mathbf{p}_u).\tag{30}
\end{equation}
\end{proof}

\section{Details of Group Fairness Results}
In this part, we measure fairness based on the classification accuracy of each sensitive attribute. Then, we show the results of our proposed model on group fairness measures.

\subsection{Statistical Parity}
For all group fairness based metrics, statistical group parity is widely used to measure the predicted rating discrepancy for binary valued sensitive attribute~\cite{bose2019compositional,yao2017beyond}. Correspondingly, we measure the statistical parity of binary attribute~(i.e., gender) in recommendation as: $1 /N \sum_{v=1}^{N}\left\|E_{u \in \operatorname{male}}\left[\hat{r}_{uv}\right]-E_{u \in \text {female}}\left[\hat{r}_{uv}\right]\right\|$. For attributes with multiple values, we borrow the idea of statistical parity and bin users into different groups based on the different attribute values. Then, we take the standard deviation of predicted ratings of each user group to measure statistical parity.

\subsection{Equal Opportunity}
Besides statistical group parity, equal opportunity is also a widely used group fairness metric. Equalized opportunity advances statistical group fairness by considering the parity of prediction accuracy of each group~\cite{NIPS2016equality}.  The equal opportunity measures group fairness of binary attribute~(i.e., gender) as:

\begin{footnotesize}
	\begin{equation*} 
	1 /N \sum_{v=1}^{N} \left\|E_{u \in \operatorname{male}}\left[||\hat{r}_{uv}-r_{uv}||\right]-E_{u \in \text {female}}\left[||\hat{r}_{uv}-r_{uv}||\right]\right\|.
	\end{equation*}
\end{footnotesize}

For attributes with multiple values, we use the idea of equal opportunity for binary values of attribute, and take the standard deviation of equal opportunity of each user group to measure group fairness. We only list our proposed \shortname under GCN as it shows better performance under PMF. Besides, the performance on Non-parity is only calculated for the binary attributes.

As shown in Figure~\ref{fig:group_fair1} and Figure~\ref{fig:group_fair2}, our proposed model achieves the best results for binary gender attribute. Our proposed model also reaches the best results for the two group based fairness metrics on Lastfm-360K dataset. However, our proposed model could not perform the best for sensitive attributes with multiple attributes on MovieLens under the equal opportunity metric. We guess a possible reason is that, the adversarial training process of \shortname relies on sufficient user data for training. As MovieLens is much smaller than Lastfm-360K, the performance drops when the attribute has multiple values with limited records.

\end{document}